\RequirePackage{amsthm,amsmath,amsfonts,amssymb}
\documentclass[a4paper,10pt]{article}
\usepackage[utf8]{inputenc}
\usepackage{float}
\usepackage{eucal}
\usepackage{graphicx}
\usepackage{amssymb}
\usepackage{amsmath}
\usepackage{amsthm}
\usepackage{makecell}
\usepackage{booktabs,caption}
\usepackage[flushleft]{threeparttable}
\usepackage[title]{appendix}

\theoremstyle{plain}
\newtheorem{theorem}{Theorem}
\newtheorem{prop}[theorem]{Proposition}
\newtheorem{lemma}[theorem]{Lemma}
\newtheorem{cor}[theorem]{Corollary}

\theoremstyle{definition}

\newtheorem{defi}{Definition}

\newcommand{\inlinesection}[1]{\textbf{{#1}.}}
\newcommand{\nnline}{\vspace{5mm}}

\def \A {\mathbf{A}}

\def \SSigma {\mathbf{\Sigma}}
\def \PPi {\mathbf{\Pi}}
\def \NN {\mathcal{N}}
\usepackage{dsfont}

\def \p {\mathbf{p}}

\def \LLambda {\mathbf{\Lambda}}
\def \y {\mathbf{y}}

\def \u {\mathbf{u}}

\def \GGamma {\mathbf{\Gamma}}
\def \grad {\nabla}
\def \X {\mathbf{X}}

\def \p {\mathbf{p}}

\def \Z {\mathbf{Z}}

\def \x {\mathbf{x}}
\def \R {\mathbb{R}}
\def \E {\mathbb{E}\hspace{0.5mm}}

\def \B {\mathbf{B}}

\def \C {\mathbf{C}}

\def \I {\mathbf{I}}

\def \U {\mathbf{U}}
\def \Var {\text{Var}}
\def \Cov {\text{Cov}}
\def \S {\mathbf{S}}
\def \tr {\text{tr}\hspace{0.5mm}}

\def \LLambda {\mathbf{\Lambda}}

\newcommand{\argmin}[1]{\underset{#1}{\arg\min}\hspace{2mm}}

\def \tworeg {\hat\beta_{\text{2REG}}}
\def \ols {\hat\beta_{\text{OLS}}}

\begin{document}


\title{Two-Stage Regularization of Pseudo-Likelihood Estimators with Application to Time Series}
\author{Erez Buchweitz\footnote{Corresponding author; erezmb@gmail.com.}, Shlomo Ahal, Oded Papish, Guy Adini}
\date{Istra Research Ltd.}
\setlength{\parindent}{0cm}

\maketitle
\thispagestyle{empty}

\section*{Abstract}

Estimators derived from score functions that are not the likelihood are in wide use in practical and modern applications. Their regularization is often carried by pseudo-posterior estimation, equivalently by adding penalty to the score function. We argue that this approach is suboptimal, and propose a two-staged alternative involving estimation of a new score function which better approximates the true likelihood for the purpose of regularization. Our approach typically identifies with maximum a-posteriori estimation if the original score function is in fact the likelihood.
We apply our theory to fitting ordinary least squares (OLS) under contemporaneous exogeneity, a setting appearing often in time series and in which OLS is the estimator of choice by practitioners.

\newpage



\section{Introduction}

Statistical methodology favors obtaining parametric estimators by maximizing the likelihood function, yet in many cases this is intractable due to practical concerns, and pseudo-likelihood score functions are used \cite{ventura}. Regularization of maximum pseudo-likelihood (pseudo-ML) estimators is often carried by combining the pseudo-likelihood with a prior \cite{sisson, ventura}, or equivalently by adding penalty to the score function, e.g.\ $L^1,L^2$ regularization in neural networks \cite{goodfellow} and in gradient boosting machines \cite{gbm}, and soft-margin support vector machines \cite{elements}. We argue that when dealing with pseudo-ML estimators, this approach of regularization is suboptimal. Instead, we suggest to combine the prior with a different pseudo-likelihood function than that which was used for obtaining the unregularized estimator. We call our approach two-stage regularization (2REG).\nnline

Let $\X$ be a random sample drawn from a distribution parameterized by $\beta$, and suppose $\hat\beta(\X)$ is an estimator of $\beta$. A 2REG estimator is of the form
\begin{align}\label{intro_2reg}
    \tworeg \ = \ {\arg\max}_{\beta} \big\{\ p(\hat\beta|\beta) \cdot p(\beta) \ \big\}
\end{align}
where $p(\hat\beta|\beta)$ typically is the density of a normal distribution $\hat\beta|\beta \sim \NN(\beta, \hat\C)$. The density $p(\hat\beta|\beta)$ acts as a pseudo-likelihood function, and $\hat\C$ is an estimator of the covariance matrix $\C=\Cov(\hat\beta)$. The density of the prior distribution is denoted $p(\beta)$. We omit from the notation the conditioning on appropriate ancillary statistics $\cite{lehmann}$. The justification for this definition is given here in short, and is elaborated upon later. Equation (\ref{intro_2reg}) applies regularization to $\hat\beta$ by combining its distribution, rather than the likelihood, with a prior. For typical maximum likelihood (ML) estimators, this approach identifies with maximum a-posteriori (MAP) estimation. The distribution of a pseudo-ML estimator, preferably consistent and asymptotically normal \cite{ventura}, is often tractable even if the likelihood function is not. This provides for a pseudo-likelihood function that gives an improved approximation of the true likelihood, specialized for the purpose of regularization of $\hat\beta$. We show that the estimator in equation (\ref{intro_2reg}) upholds some desirable properties, and is furthermore standard and straightforward to compute. Indeed, the covariance estimate $\hat\C$ can be obtained by bootstrap (or otherwise), and the optimization in equation (\ref{intro_2reg}) reduces in the general case to an ordinary least squares (OLS) fit. The theory of 2REG regularization is closely related to the synthetic likelihood approach of Wood \cite{wood} (see also \cite{drovandi}) and to the body of work originating from Newey and West \cite{newey_west} in econometrics (see also \cite{zeileis}). 
\nnline

We accompany our discussion of 2REG regularization throughout the paper with an application to fitting linear models in contemporaneous exogeneity setting \cite{greene}. This example is analytically tractable for its better part, and we explain in rigorous terms why an inefficient OLS estimator is preferred over more evident approximations of the true likelihood; what its caveats are and how 2REG compensates for them, whereas standard regularization techniques do not; how to obtain the covariance estimator $\hat\C$. The application to contemporaneous exogeneity bears additional independent interest, as this setting is prevalent in time series problems, and OLS is widely used by practitioners whenever strict exogeneity \cite{greene} is not granted \cite{angrist, cameron, wooldridge, zeileis}.
\nnline

An estimate of the covariance matrix $\C$ is required for computation of the estimator in equation (\ref{intro_2reg}), and we emphasize that the benefit in using 2REG regularization is mitigated by the estimation error of the covariance. As a generic and robust approach, we suggest to estimate $\C$ by the sample covariance generated by plug-in bootstrap estimates. Altenatively, for a given problem one may devise a specialized covariance estimator; consider for example the semiparametric family of HAC estimators of the OLS covariance \cite{zeileis}. Having obtained a crude estimate, it is well established in the literature that covariance estimators require regularization (see \cite{warton} and the references therein). We use a prior covariance matrix based on the original pseudo-likelihood function, and follow the convex combination approach laid by Ledoit and Wolf \cite{ledoit_wolf} and Warton \cite{warton} with some changes. A final act of normalization is then applied to settle a three-step method of estimating the covariance matrix, comprising estimation, regularization and normalization.\nnline

The remainder of this paper is organized into three major parts. In section \ref{section_time_series} we introduce the time series problem of fitting a linear model under contemporaneous exogeneity, which will guide us throughout the paper as the leading example. In section \ref{section_2reg} we introduce 2REG regularization, we apply it to the time series problem and derive its properties in general setting. In section \ref{section_cov} we discuss estimation of the covariance matrix required for computation of the 2REG estimator. In the remaining sections \ref{section_sim} and \ref{section_real_data} we demonstrate the performance of 2REG in a suite of simulation studies, and on a real-world data set of NYSE stock prices.

\section{Time Series Problem}\label{section_time_series}

Suppose $\y=\X\beta+\varepsilon$, where $\X$ is a $n\times p$ random matrix of rank $p$ of covariate observations, $\beta\in\R^p$ is fixed, and $\varepsilon,\y\in\R^n$ are random vectors. We treat any vector as a column vector. Under \textit{contemporaneous exogeneity} \cite{greene, wooldridge}, we assume the residual noise $\varepsilon$ upholds
\begin{align}\label{section_time_series_contemporaneous}
\E[\varepsilon_i|\X_{i1},..,\X_{ip}] \ = \ 0
\end{align}
for all $i=1,..,n$. Here, the noise is required to have mean zero conditioned upon the contemporaneously sampled covariates. This relaxes the stronger assumption of \textit{strict exogeneity}, i.e.\  $\E[\varepsilon|\X]=0$, by permitting covariates to contain information on the value of $\varepsilon$ between observations. Consider, for example, time shifts of $\y$ as covariates \cite{greene}. The OLS estimator $\ols  =  (\X^T\X)^{-1}\X^T\y$
is obtained by maximizing the pseudo-likelihood which assumes 
\begin{align}\label{section_time_series_pseudo_ols}
\varepsilon|\X\  \sim\  \NN(0,\sigma^2\I)
\end{align}
for some $\sigma>0$, where $\I$ is the identity matrix in appropriate dimension. The OLS pseudo-likelihood misspecifies the true likelihood in at least two meaningful respects. First, it asserts strict exogeneity $\E[\varepsilon|\X]=0$, a strong assumption we explicitly circumvented. The alternative requires modeling the elusive distribution of $\X$ into the likelihood. Second, it assumes the noise observations are uncorrelated and homoscedastic.\nnline

Despite its misgivings, a simple argument shows that an OLS estimator is consistent in this setting, under mild assumptions. Indeed, $\ols=\beta+(\X^T\X)^{-1}\X^T\varepsilon$ while contemporaneous exogeneity (\ref{section_time_series_contemporaneous}) entails that $\X^T\varepsilon/n$ is vanishing as $n\to\infty$, from which consistency follows \cite{wooldridge}. Not so, if we were to take into account correlation and heteroscedasticity within $\varepsilon$. Indeed, if $\varepsilon|\X\sim\NN(0,\sigma^2\SSigma)$ for some covariance matrix $\SSigma$, the maximum-likelihood estimator is generalized least squares (GLS) \cite{agresti}, given by applying a linear transformation to unravel the autocorrelation structure, then solving using OLS;
\begin{align}\label{section_time_series_transform}
\S\y \ = \ \S\X\beta + \S\varepsilon
\end{align}
where $\S$ is a $n\times n$ matrix such that $\Cov(\S\varepsilon)=\sigma^2\I$. A non-diagonal $\S$ does not in general preserve contemporaneous exogeneity, which is required to guarantee consistency of the subsequent OLS applied to the transformed model equations (\ref{section_time_series_transform}). Thus, consistency of GLS is not granted and as result many authors and practitioners refrain from using it (e.g.\ \cite{angrist, cameron, wooldridge}), opting instead for the consistent OLS estimator. The popularity of Newey-West-type standard errors \cite{zeileis} in econometrics is a further testament to this. We make two further notes. First, a non-diagonal $\S$ is mandated by correlation within $\varepsilon$, a setting appearing frequently in time series where one is resorted to modeling multiple steps into the future, causing overlap and correlation within $\varepsilon$, or facing compounded misspecification error \cite{jorda}. Second, even if the unconditional distribution $\varepsilon\sim(0,\SSigma)$ is known, the more intricate conditional distribution $\varepsilon|\X$ might be very different.\nnline

Although robust, the OLS pseudo-likelihood can be shown to overlook certain aspects that might be prevalent in the data and affect the fit. For instance, covariates with correlated observations, or with random effects endure increased variance of fit. Since covariates within the same model differ from each other in these respects, the produced fit might be adversely dominated by high-variance, uncertain elements. The problem is furthermore aggravated in that the difficulty in learning using heavily autocorrelated or random-effect covariates might in practice deter us from including such covariates in our model, or unwittingly fail if we nevertheless try. We demonstrate via examples.
\nnline

We first state a lemma relating the asymptotic variance of the OLS estimator to the joint covariance structure of a univariate $\X$ and $\varepsilon$. The conditions for lemma \ref{section_time_series_lemma} are cumbersome but mild and amount to well-behavior of the data as the sample size $n$ increases (see \cite{wooldridge}).

\begin{lemma}\label{section_time_series_lemma}
    Let $\{\X_n\}_{n=1}^\infty$ be a sequence of random variables with mean zero and uniformly bounded fourth moments, and let  $\{\varepsilon_n\}_{n=1}^\infty$ be another sequence of random variables with mean zero conditioned upon any finite subset of $\{\X_n\}_{n=1}^\infty$ and with finite variances.
    Denote by $\{\SSigma_n\}_{n=1}^\infty$ the conditional finite-sample noise covariance matrices, i.e.\ $(\SSigma_n)_{ij}=\Cov(\varepsilon_i,\varepsilon_j|\X_1,..,\X_n)$ for $i,j=1,..,n$.
    Assume that:
    \begin{enumerate}
        \item $(\X_1^2+...+\X_n^2)/n\to 1$ in probability, as $n\to\infty$.
        \item $\{\sum_{i,j=1}^n(\SSigma_n)_{ij}\X_i\X_j/n\}_{n=1}^\infty$ converges in probability to a constant, as $n\to\infty$.
        \item $\{\SSigma_n/n\}_{n=1}^\infty$ are uniformly bounded in operator norm.
        \item $\big\{n(\X_1^2+...+\X_n^2)^{-1}\big\}_{n=1}^\infty$ are uniformly integrable for large enough $n$.
    \end{enumerate}
    
    Suppose $\y_n  =  \X_n\beta+\varepsilon_n$ for all $n$, where  $\beta\in\R$ is fixed, and denote the OLS estimator for the finite sample of size $n$ by  $\hat\beta_{n}=(\sum_{i=1}^n\X_i^2)^{-1}(\sum_{i=1}^n\X_i\y_i)$. Then, 
    $$
    \lim_{n\to\infty} n\Var(\hat\beta_n) \ = \ \lim_{n\to\infty} \frac{1}{n}\sum_{i,j=1}^n \Cov(\X_i\varepsilon_i,\ \X_j\varepsilon_j).
    $$
    provided that the limit on the right-hand side exists.
\end{lemma}

The proof of lemma \ref{section_time_series_lemma} appears in appendix \ref{appendix_time_series}. The lemma readily generalizes to a multivariate setting.  
\nnline

Our first example is that of an autocorrelated covariate, i.e.\ whose observations are correlated. Lemma \ref{section_time_series_lemma} shows that when the entries of $\varepsilon$ are uncorrelated, autocorrelation in $\X$ does not come into effect. However, when $\varepsilon$ is itself autocorrelated, the OLS fit has higher variance the heavier the autocorrelation in $\X$ is. Over data containing covariates with varying degrees of autocorrelation, the OLS fit will admit uneven variance. We take,   as example, autoregressive $\X$ and $\varepsilon$.

\begin{prop}\label{section_time_series_prop_ar}
    In the setting of lemma \ref{section_time_series_lemma}, if
    $$
    \Cov(\X_i,\X_j) \ = \ \pi^{|i-j|} \ \ \ \ ;  \ \ \ \Cov(\varepsilon_i,\varepsilon_j|\{\X_n\}_{n=1}^\infty) \ = \ \sigma^2\rho^{|i-j|}
    $$
    for $\sigma>0$ and $0\leq \pi,\rho<1$, then
    $$
    n\Var(\hat\beta_n) \ \overset{n\to\infty}{\longrightarrow} \ \sigma^2\cdot\frac{1+\pi\rho}{1-\pi\rho}.
    $$
\end{prop}

Fixing $\rho>0$ the autocorrelation coefficient of $\varepsilon$, we see that OLS variance is inflated along with $\pi$ the autocorrelation coefficient of $\X$. The proof of proposition \ref{section_time_series_prop_ar} is given in appendix \ref{appendix_time_series}.
\nnline

We move on to discuss random effects, which by standard argument reduce to autocorrelation and heteroscedasticity in the noise with magnitude aligned to the covariates associated with random effects. Suppose $\y_i  =  \X_i\beta_i + \varepsilon_i$ for all $i=1,..,n$, with $\beta_1,..,\beta_n$ random such that $\E[\beta_i|\X]=\beta$ is constant among all $i$. An equivalent formulation is
\begin{align}
    \y \ = \ \X\beta + \eta \ \ \ \ ; \ \ \ \ \eta_i \ = \ \varepsilon_i + \X_i(\beta_i-\beta)
\end{align}
for all $i=1,..,n$. We have $\E[\eta|\X] = 0$ so long as $\E[\varepsilon|\X] = 0$, yet the noise variance is dependent of the covariate value. Standard methods for estimating the fixed and random effects in a mixed effects model are similar to GLS in their utilization of the noise and effect covariances \cite{agresti}. They are therefore likewise susceptible to inconsistency under contemporaneous exogeneity, and we are resorted to using OLS to estimate the effects.
\nnline

We demonstrate that fitting a covariate associated with a random effect is accompanied with increased variance, in a simplistic setting of an autoregressive random effect.

\begin{prop}\label{section_time_series_prop_random1}
    Suppose $\y_i=\X_i\beta+(\varepsilon_i+\X_ib_i)$, where $\{\X_i\}_{i=1}^\infty$ are independent and identically distributed with mean zero and variance one, $\{\varepsilon_i\}_{i=1}^\infty$ are independent with mean zero and variance $\sigma^2$, and $\{b_i\}_{i=1}^\infty$ have mean zero and covariance $\Cov(b_i,b_j)=\Var(b_1)\tau^{|i-j|}$. Assume furthermore that $\{\X_i\}_{i=1}^\infty$, $\{\varepsilon_i\}_{i=1}^\infty$ and $\{b_i\}_{i=1}^\infty$ are independent of each other, and that the conditions for lemma \ref{section_time_series_lemma} are met. Then, in the setting of lemma \ref{section_time_series_lemma} it holds that
    $$
    n\Var(\hat\beta_n) \ \overset{n\to\infty}{\longrightarrow} \   \sigma^2 + \Var(b_1)\cdot\bigg(\E[\X_1^4] \ + \ \frac{2\tau}{1-\tau}\bigg).
    $$
\end{prop}

Note that $\E[\X_1^4]\geq \E[\X_1^2]^2=1$ by Jensen's inequality. See that the OLS variance grows as the variance and autocorrelation of the random effect grow. The proof of proposition \ref{section_time_series_prop_random1} is given in appendix \ref{appendix_time_series}. To complement proposition \ref{section_time_series_prop_random1}, we compare to a similar situation in which the variance of the noise is not aligned with the covariate.

\begin{prop}\label{section_time_series_prop_random2}
    Suppose $\y_i=\X_i\beta+(\varepsilon_i+\Z_ib_i)$, where $\{\Z_i\}_{i=1}^\infty$ are independent with mean zero and variance one, and are furthermore independent of $\{\X_i\}_{i=1}^\infty$, $\{\varepsilon_i\}_{i=1}^\infty$, $\{b_i\}_{i=1}^\infty$. Assume that the conditions of propositiom \ref{section_time_series_prop_random1} otherwise hold. Then,
    $$
    n\Var(\hat\beta_n) \ \overset{n\to\infty}{\longrightarrow} \   \sigma^2 + \Var(b_1).
    $$
\end{prop}

The proof of proposition \ref{section_time_series_prop_random2} is given in appendix \ref{appendix_time_series}.
\nnline

The introduction of the time series setting is concluded. We have shown that the OLS psuedo-likelihood overlooks autocorrelation and random effects, causing a pattern of uneven variance in the different elements of the OLS fit. Regularization approaches that employ the same pseudo-likelihood, e.g.\ standard ridge and lasso, will likewise be ignorant to variance patterns stemming from autocorrelation and random effects. In contrast, 2REG regularization explicitly takes into account the covariance structure of the underlying estimator, allowing for a better-informed combination with the prior distribution. We now turn to formally defining the 2REG estimator.

\section{2REG Regularization}\label{section_2reg}

As before, let $\X$ be a random sample drawn from a distribution parameterized by $\beta$, let $\hat\beta(\X)$ be an estimator of $\beta$, and let $p(\beta)$ be the density of a prior distribution of $\beta$. We make the following definitions.
\begin{defi}\label{section_2reg_defi}
    A 2REG estimator is defined by $\tworeg = {\arg\max}_{\beta} \{p(\hat\beta|\beta) \cdot p(\beta) \}$, where $p(\hat\beta|\beta)$ is some function of $\hat\beta,\beta$ and possibly ancillary statistics. 
    
    If $p(\hat\beta|\beta)$ is the probability density of observing $\hat\beta$ given an underlying parameter value $\beta$, we say that $\tworeg$ is a \textit{correctly specified} 2REG estimator. 
    
    If $p(\hat\beta|\beta)$ is the density of the normal distribution $\hat\beta|\beta\sim\NN(\beta,\C)$ for some fixed covariance matrix $\C$, we say that $\tworeg$ is a \textit{normal} 2REG estimator.
\end{defi}

Many pseudo-ML estimators are consistent and asymptotically normally distributed \cite{ventura}. The OLS estimator under contemporaneous exogeneity is such \cite{wooldridge}, and furthermore $\ols-\beta$ is a pivotal quantity, in other words $\beta$ is a location parameter, hence the covariance $\Cov(\ols|\beta)$ is independent of $\beta$. Under these circumstances, a normal 2REG makes for a good approximation of the correctly specified 2REG estimator.
\nnline

To compute the normal 2REG estimator for given covariance matrix $\C$ and prior $p(\beta)$, one is reduced in the general case to solving OLS with prior $p(\beta)$.
\begin{prop}\label{section_2reg_prop_normal}
    A normal 2REG estimator is given by solving OLS with prior $p(\beta)$ over data $(\X,\y)$ such that $\X^T\X=\C^{-1}$ and $\X^T\y=\C^{-1}\hat\beta$.
\end{prop}
\begin{proof}
    A normal 2REG estimator assumes the distribution $\hat\beta|\beta\sim\NN(\beta,\C)$, consequently
    \begin{align*}
        -2\log p(\hat\beta|\beta) \ = \ (\beta-\hat\beta)^T\C^{-1}(\beta-\hat\beta) \ = \ (\y-\X\beta)^T(\y-\X\beta) + c
    \end{align*}
    where $c$ does not depend on $\beta$.
\end{proof}

Note that such $\X,\y$ may be obtained via Cholesky decomposition. We now apply normal 2REG regularization to the time series problem discussed in section \ref{section_time_series}, in the form of ridge regularization.
\nnline

\inlinesection{Application to the time series problem}

\begin{cor}\label{section_2reg_cor_ols}
    If $\hat\beta$ is an OLS estimator then the normal 2REG ridge estimator is given by
    $$
    \tworeg \ = \ (\X^T\X + \lambda\X^T\X\C)^{-1}\X^T\y
    $$
    where $\C=\Cov(\ols)$ and $\lambda>0$.
\end{cor}
\begin{proof}
    By proposition \ref{section_2reg_prop_normal} we have $\tworeg  =  (\C^{-1}+\lambda\I)^{-1}\C^{-1}\ols$, hence
    $$
    \tworeg \ = \ (\I+\lambda\C)^{-1}(\X^T\X)^{-1}\X^T\y \ = \ (\X^T\X+\lambda\X^T\X\C)^{-1}\X^T\y
    $$
    as required.
\end{proof}

Compare with the standard ridge estimator $\hat\beta=(\X^T\X+\lambda\I)^{-1}\X^T\y$, which de facto assumes $\C=(\X^T\X)^{-1}$ on account of the misspecified OLS pseudo-likelihood (\ref{section_time_series_pseudo_ols}). Standard ridge does not take into account autocorrelation and other factors that might affect OLS covariance, whereas 2REG ridge imposes penalty adjusted according to the covariance of the underlying estimator. Note that $\C$ is not in general known, and one is usually resorted to estimating it from the data.
\nnline

We accentuate suboptimality of standard ridge compared to 2REG ridge via two simple observations. First, with 2REG ridge increasing the penalty parameter is guaranteed to entail reduction in marginal variance, whereas for standard ridge this is not in general true.

\begin{prop}\label{section_2reg_prop_monotone}
    Fix $\LLambda$ a $p\times p$ matrix and define $\hat\beta_\lambda=(\X^T\X+\lambda\X^T\X\LLambda)^{-1}\X^T\y$ for any $\lambda\geq0$. Then, $\Cov(\hat\beta_{\lambda})$ is monotone decreasing in $\lambda\geq 0$ in the sense of positive definite matrices, if and only if $\LLambda^T\C^{-1}+\C^{-1}\LLambda$ is positive definite, where $\C=\Cov(\ols)$.
\end{prop}
Indeed, 2REG ridge uses $\LLambda=\C$ whereas standard ridge uses $\LLambda=(\X^T\X)^{-1}$. The proof of proposition \ref{section_2reg_prop_monotone} appears in appendix \ref{appendix_monotone}. Second, see that the statistics $(\X,\ols)$ suffice for computing the standard ridge estimator. Given that indeed only $(\X,\ols)$ are observed, 2REG ridge is in fact the MAP estimator.
\begin{prop}\label{section_2reg_prop_pmap}
    A correctly specified $\tworeg$ is the MAP estimator given that only $\hat\beta$ is observed, possibly along with ancillary statistics.
\end{prop}
\begin{proof}
    Follows immediately from definition \ref{section_2reg_defi} given that $p(\hat\beta|\beta)=p(\hat\beta,U|\beta)\cdot c$ for any ancillary statistic $U$, where $c$ is independent of $\beta$.
\end{proof}

We now turn to highlighting several further simple observations regarding 2REG estimators in general setting.\nnline

\inlinesection{General properties of 2REG}

A consequence of proposition \ref{section_2reg_prop_pmap} is that if $\hat\beta$ is sufficient or has ancillary complement, then a correctly specified $\tworeg$ is the MAP estimator. This is the case with ML estimators in exponential families.
\begin{cor}\label{section_2reg_cor}
    If $\hat\beta$ is the ML estimator in a regular minimal exponential family which is not curved, then a correctly specified $\tworeg$ is the MAP estimator.    
\end{cor}
\begin{proof}
     In a regular exponential family in minimal canonical parameterization, if a ML estimator exists it is sufficient \cite{sundberg}. Since the family is minimal and not curved, the ML estimator of the parameter $\beta$ is a bijection of the canonical ML estimator, therefore also sufficient. Since sufficiency implies Bayes sufficiency \cite{blackwell}, by proposition \ref{section_2reg_prop_pmap} the correctly specified 2REG identifies with MAP.
\end{proof}

Thus, 2REG typically identifies with MAP for ML estimators, while providing a generalization to pseudo-ML estimators. Next, we relate some of the properties of 2REG to its underlying estimator. An immediate observation is that 2REG is constant on level sets of the underlying, in other words it is a function of the underlying estimator. This implies MAP being a function of the ML estimator.

\begin{prop}\label{section_2reg_prop_level_sets}
    Let $\X,\X'$ be data sets. If $\hat\beta(\X)=\hat\beta(\X')$ then $\tworeg(\X)=\tworeg(\X')$.
\end{prop}
\begin{proof}
    The objective function in definition \ref{section_2reg_defi} is equal between $\X,\X'$.
\end{proof}

A converse of proposition \ref{section_2reg_prop_level_sets} holds for location-parameter families which are log-concave in the parameter, such as the normal 2REG.

\begin{prop}\label{section_tworeg_prop_level_sets_converse}
    Suppose $\hat\beta-\beta$ is pivotal and $p(\hat\beta|\beta)$ is strictly log-concave and differentiable as a function of $\beta$. Assume furthermore that the prior density $p(\beta)$ is strictly positive everywhere and differentiable, and let $\X,\X'$ be data sets. If $\tworeg(\X)=\tworeg(\X')$ then $\hat\beta(\X)=\hat\beta(\X')$.
\end{prop}

\begin{proof}
    Denote $f_{\hat\beta}(\beta)= p(\hat\beta|\beta)$ and we may write
    \begin{align}\label{section_level_prop_converse_eq1}
        \tworeg(\hat\beta) \ = \ {\arg\max}_{\beta} \ \{ \ f_{\hat\beta}(\beta)\cdot p(\beta) \ \}.
    \end{align}
    Suppose that $\tworeg(\hat\beta)=\tworeg(\hat\beta')$ and we are required to show $\hat\beta=\hat\beta'$. By pivotality we have  $f_{\hat\beta+b}(\beta)=f_{\hat\beta}(\beta-b)$, therefore
    \begin{align}\label{section_level_prop_converse_eq2}
        \tworeg(\hat\beta') \ = \ {\arg\max}_{\beta} \ \{ \ f_{\hat\beta}(\beta-b)\cdot p(\beta) \ \}
    \end{align}
    where $b=\hat\beta'-\hat\beta$. Differentiating the log of the objective functions in equations (\ref{section_level_prop_converse_eq1}) and (\ref{section_level_prop_converse_eq2}) at the point $\beta^*$ at which they are both maximized, we get
    \begin{align*}
        \grad \log f_{\hat\beta}(\beta^*) + \grad \log p(\beta^*) \ = \ 0 \ = \ \grad \log f_{\hat\beta}(\beta^*-b) + \grad \log p(\beta^*),
    \end{align*}
    consequently $\grad \log f_{\hat\beta}(\beta^*)=\grad \log f_{\hat\beta}(\beta^*-b)$. A strictly concave function such as $\log f_{\hat\beta}$ has strictly monotone gradient which cannot equate at distinct points, hence necessarily $b=0$ meaning $\hat\beta=\hat\beta'$.
\end{proof}

One further property of the underlying parameter the normal 2REG preserves is consistency. This is particularly important for our application to time series, as our original motivation was to develop a more efficient estimator while still guaranteeing consistency, recall section \ref{section_time_series}.

\begin{prop}\label{section_2reg_prop_consistent}
    Suppose the prior density $p(\beta)$ is strictly positive everywhere, continuous and bounded. Then, if $\hat\beta$ is consistent and the covariance $\C$ vanishes asymptotically, then normal 2REG is consistent.
\end{prop}
\begin{proof}
     Let $\{\hat\beta_n\}_{n=1}^\infty$ be a sequence of instances of the underlying estimator such that $\hat\beta_n\to \beta^*$ in probability. Let $\{\hat\gamma_n\}_{n=1}^\infty$ be a corresponding sequence of normal 2REG estimators assuming $\hat\beta_n|\beta \sim  \NN(\beta,\C_n)$
     with $\C_n\to 0$ as $n\to \infty$. To show $\hat\gamma_n \to \beta^*$ in probability, fix $\delta,\varepsilon>0$ and we find $N$ such that for all $n\geq N$ it holds that $\|\hat\gamma_n-\beta^*\|<2\varepsilon$ with probability at least $1-\delta$, where $\|\beta\|^2=\beta^T\beta$. Indeed, define the local maximal \textit{leverage} of the prior
     \begin{align*}
     \kappa \ = \ \frac{\sup_{\beta} \{ p(\beta) \}}{\inf_{\beta\in B_{\varepsilon(\beta^*)}} \{ p(\beta) \}}
     \end{align*}
     where the closed ball with radius $\varepsilon$ centered at $\beta^*$ is denoted by $B_{\varepsilon}(\beta^*)=\{\beta:\|\beta-\beta^*\|\leq\varepsilon\}$. See that $\kappa\geq 1$ and since $p(\beta)$ is strictly positive, continuous and bounded $\kappa$ is finite. By converegence in probability of $\hat\beta_n$, let $N_1$ be such that for all $n\geq N_1$ we have $\|\hat\beta_n-\beta^*\|<\varepsilon$ with probability at least $1-\delta$. We furthermore bound the operator norm of $\C_n$. By convergence of $\C_n$ to zero, let $N_2$ be such that for all $n\geq N_2$ and all $\beta$ we have $\|\C_n\beta\|\leq \alpha\|\beta\|$ where $\alpha=\varepsilon^2/2\log\kappa$. As $\C_n$ is symmetric and positive definite, this entails $\beta^T\C_n^{-1}\beta\geq \alpha^{-1}\|\beta\|^2$ for all $\beta$. Denote $f_n(\beta)=\exp\{-(\hat\beta_n-\beta)^T\C_n^{-1}(\hat\beta_n-\beta)/2\}$ and see that
     \begin{align*}
         \hat\gamma_n \ = \ \arg{\max}_\beta\  \{ f_n(\beta)\cdot p(\beta) \}.
     \end{align*}
     Suppose $n\geq \max\{N_1,N_2\}$, then with probability at least $1-\delta$ we have $\hat\beta_n\in B_\varepsilon(\beta^*)$ which we use to show \begin{align}\label{section_remap_prop_consistency_target}
         f_n(\hat\beta_n)\cdot p(\hat\beta_n)\ \geq \ f_n(\beta)\cdot p(\beta)
     \end{align} 
     for any $\beta\notin B_\varepsilon(\hat\beta_n)$. This would entail $\hat\gamma_n\in B_\varepsilon(\hat\beta_n)\subseteq B_{2\varepsilon}(\beta^*)$ thus concluding the proof of the proposition. Indeed, $f_n(\hat\beta_n)=1$ while for $\beta\notin B_\varepsilon(\hat\beta_n)$ we have
     \begin{align*}
         -2\log f_n(\beta) \ = \ (\hat\beta_n-\beta)^T\C_n^{-1}(\hat\beta_n-\beta) \ \geq \ \alpha^{-1}\|\hat\beta_n-\beta\|^2 \ \geq \ \alpha^{-1}\varepsilon^2
     \end{align*}
     consequently $f_n(\beta)\leq e^{-\varepsilon^2/2\alpha}$. On the other hand, $p(\beta)/p(\hat\beta_n)\leq \kappa$ by definition of $\kappa$, and we combine to get
     \begin{align*}
         \frac{f_n(\hat\beta_n)}{f_n(\beta)} \ \geq \ e^{\varepsilon^2/2\alpha} \ = \ \kappa \ \geq \ \frac{p(\beta)}{p(\hat\beta_n)}
     \end{align*}
     which entails the desired inequality (\ref{section_remap_prop_consistency_target}) and the proposition is proven.
\end{proof}

We thus conclude our discussion of the properties of 2REG estimators, and divert our attention to estimation of the covariance $\C$.

\section{Covariance Estimation}\label{section_cov}

An estimator of the covariance matrix $\C$ is required for the computation of the normal 2REG estimator, recall definition \ref{section_2reg_defi}. We describe an approach comprising three steps; estimation, regularization and normalization. First, we obtain a crude estimate of the covariance, for example via bootstrap. Indeed, let $\hat\beta^1,..,\hat\beta^B$ be a host of bootstrap plug-in estimates, and use the sample covariance
\begin{align*}
    \hat\C \ = \ \frac{1}{B-1}\sum_{b=1}^B (\hat\beta^b-\bar{\beta})(\hat\beta^b-\bar{\beta})^T
\end{align*}
where $\bar{\beta}=(\hat\beta^1+...+\hat\beta^B)/B$. We note that when resampling data with correlated observations, such as time series, it is advisable to sample continuous blocks \cite{cameron, kreiss}. Alternatively, for a given problem one may devise a specialized covariance estimator. For instance, estimation of the covariance of an OLS estimator, as our time series problem requires, has been studied in the econometric literature under the name of heteroscedasticity and autocorrelation consistent (HAC) estimation, originating from the work of Newey and West \cite{newey_west}. We refer to the introduction of HAC estimators by Zeileis \cite{zeileis} and the references therein. We describe here a simple and robust resampling procedure, akin to Colin Cameron and Miller \cite{cameron}.
\nnline

Partition the data into continuous folds $(\X_1,\y_1),..,(\X_\Omega,\y_\Omega)$. For each fold $\omega=1,..,\Omega$, the OLS estimator $\hat\beta^\omega$ is computed on the combined data of all the folds, apart from the $\omega$-th fold $(\X_\omega,\y_\omega)$. The out-of-sample residual for the $\omega$-th fold is then $\hat\varepsilon_{\omega}=\y_\omega-\X_\omega\hat\beta^\omega$. Use the covariance estimator
\begin{align*}
    \hat\C \ = \ (\X^T\X)^{-1}\bigg(\sum_{\omega=1}^\Omega \X_\omega^T\hat\varepsilon_\omega\hat\varepsilon_\omega^T\X_\omega\bigg)(\X^T\X)^{-1}
\end{align*}
relying on the observation $\ols =\beta+(\X^T\X)^{-1}\X^T\varepsilon$. This resembles a HAC estimator, with out-of-sample residuals obtained by cross validation, and a Lumley-Heagerty-type kernel \cite{zeileis}. We postpone discussing regularization of covariance estimators, and skip directly to the concluding step of normalization.

\nnline

\inlinesection{Normalization}

Recall that the 2REG ridge estimator given in corollary \ref{section_2reg_cor_ols} replaces the penalty matrix $\I$ appearing in standard ridge with $\X^T\X\hat\C$. However, the covariance implied by the OLS pseudo-likelihood (\ref{section_time_series_pseudo_ols}) is in fact $\C=\sigma^2(\X^T\X)^{-1}$, hence standard ridge absorbs the noise variance $\sigma^2$ into the penalty parameter $\lambda$. In order for 2REG ridge to follow suit we normalize the penalty matrix $\X^T\X\hat\C$ by its mean diagonal element, that is we use
$$
    \hat\C_{\text{norm}} \ = \ \hat\C \cdot p /  \tr(\X^T\X\hat\C).
$$
In this case the penalty matrix upholds $\tr(\X^T\X\hat\C_{\text{norm}})=p$ akin to standard ridge. This normalization is indifferent to rescaling of the covariates in $\X$. We now turn to regularization of the covariance estimator.
\nnline

\inlinesection{Regularization}

We detail two approaches of regularizing the covariance estimator, which we ultimately combine; a posteriori estimation and denoising by principal component analysis (PCA). These follow the convex combination approach of Ledoit and Wolf \cite{ledoit_wolf} and Warton \cite{warton}, with a number of changes. First, Warton applies shrinkage to the correlation matrix, leaving the diagonal covariance elements unscathed. In our setting, the variances of the individual entries of the underlying estimator are meaningful, and we wish to regularize them as well, though distinctly from the co-variance components. The second consideration regards to the choice of a prior covariance matrix $\PPi$. In lack of an evident prior, one may use the identity matrix as in Ledoit and Wolf, and Warton. In our case, a natural candidate arises; the covariance implied by the original pseudo-likelihood function. In the OLS case (\ref{section_time_series_pseudo_ols}), this is
\begin{align*}
\PPi \ = \ (\X^T\X)^{-1}\cdot \tr(\hat\C)/\tr((\X^T\X)^{-1})
\end{align*}
scaled to have $\tr(\PPi)=\tr(\hat\C)$, as the sample covariance trace is unbiased \cite{ledoit_wolf}.
\nnline

\inlinesection{A posteriori estimation}

We use a convex combination of the crude estimator and the prior \cite{ledoit_wolf}, as alleged by the following proposition.
\begin{prop}\label{section_cov_prop_convex}
    The following methods are equivalent for combining $\hat\C$ with $\PPi$ as prior.
\begin{enumerate}
    \item A convex combination: $(1-\kappa)\hat\C+\kappa\PPi$ where $0\leq \kappa \leq 1$.
    \item Assuming a Wishart distribution on $\hat\C$ and an inverse-Wishart prior with natural parameter proportional to $\PPi$.
    \item A ridge-type prior: \ $\arg\min_\GGamma \big\{\|\hat\C-\GGamma\|_F^2 + \lambda \|\GGamma-\PPi\|_F^2\big\}$ where $\lambda=\kappa/(1-\kappa)$ and $0\leq\kappa\leq 1$.
\end{enumerate}
\end{prop}
    The Frobenius norm is given by $\|\GGamma\|_F^2=\tr(\GGamma^T\GGamma)$. 
\begin{proof}[Proof of proposition \ref{section_cov_prop_convex}]
    The inverse-Wishart distribution is conjugate to the Wishart distribution, and their combination is linear in the parameters $\hat\C$ and $\PPi$ \cite{anderson}. Equivalence to the ridge-type prior can be derived by differentiation.
\end{proof}

We thus combine the empirical estimate $\hat\C$ with a prior by
\begin{align}\label{section_estimation_regularization_combination}
    \hat\C_\kappa \ = \ (1-\kappa)\hat\C \ + \ \kappa\PPi
\end{align}
where $0\leq\kappa\leq 1$. When $\kappa=1$, 2REG ridge as in corollary \ref{section_2reg_cor_ols} degenerates back into standard ridge. Thus, the parameter $\kappa$ serves as both to mitigate the degrees of freedom of the covariance estimator, and to scale between standard ridge and 2REG ridge.
\nnline

\inlinesection{PCA denoising}

In extreme cases, the covariance may be so difficult to estimate reasonably, leading to a 2REG estimator too heavily inclined toward the prior $\PPi$. However, the degrees of freedom can be reduced if we use the prior for denoising $\hat\C$ via principal component analysis (PCA). 
\nnline

Indeed, let $\U$ be an orthogonal matrix such that $\U^T\PPi\U$ is diagonal. Then the columns of $\U$ are an orthogonal basis of eigenvectors of $\PPi$. Define the orthogonal projection onto the set of eigenvectors of $\PPi$ by
\begin{align*}
    p_{\PPi}(\hat\C) \ = \ \U\big((\U^T\hat\C\U)\circ\I\big)\U^T
\end{align*}
where $\circ$ is the Hadamard entry-wise product, $(\A\circ\B)_{ij}=\A_{ij}\B_{ij}$. The projection $p_\PPi(\hat\C)$ switches $\hat\C$ to the basis $\U$, eliminates all off-diagonal entries, then switches back to the original basis. This is in fact an orthogonal projection in the space of matrices; $\p_\PPi(\hat\C)$ is the closest symmetric matrix to $\hat\C$ which has the same eigenvectors as $\PPi$, in Frobenius norm. This allows us to penalize the off-diagonal elements in the PCA basis, formally
\begin{align*}
    \argmin{\GGamma} \Big\{ \  \|\hat\C-\GGamma\|_F^2+\lambda\sum_{i\neq j}(\u_i^T\GGamma\u_j)^2 \ \Big\} \ = \ (1-\mu)\hat\C + \mu\cdot  p_\PPi(\hat\C)
\end{align*}
where $\u_1,..,\u_p$ are the columns of $\U$ and $\mu=\lambda/(1+\lambda)$. We combine the two approaches for regularization arriving at the regularized estimator
\begin{align}\label{section_cov_mu_kappa}
    \hat\C_{\mu,\kappa} \ = \ (1-\kappa)\big((1-\mu)\hat\C\ + \ \mu\cdot p_\PPi(\hat\C) \big) \ + \ \kappa \PPi.
\end{align}

\inlinesection{Choosing $\kappa$ and $\mu$}

In order to choose appropriate values $\kappa,\mu$ for a given problem, we score how well they predict the covariance matrix out-of-sample using resampling, for instance with cross validation (see reference within \cite{warton}). To this end, a metric $d(\C,\C')$ between covariance matrices is required. Let $\hat\C$ be the covariance estimate computed on a part of the data, and let $\hat\C_{\text{out}}$ be computed on an independent part of the data. Then, select
\begin{align*}
    \kappa^*,\mu^* \ = \ \argmin{\kappa,\mu} d\big(\hat\C_{\mu,\kappa},\  \hat\C_{\text{out}}\big)
\end{align*}
over many iterations of the resampling algorithm. The metric $d$ can be any norm operating on covariance matrices, for instance a matrix norm, or, considering that $\hat\C$ is assumed to be the covariance of a normal distribution, it may use any metric $d'$ between probability distributions
$$
d(\C, \C') \ = \ d'\big(\NN(\beta,\C),\ \NN(\beta,\C')\big).
$$

\section{Simulation Study}\label{section_sim}

In this section we describe the results of two simulation studies examining the performance of the 2REG ridge estimator in different settings. \nnline

In the first study, one of the covariates has significantly higher autocorrelation than the others. Let $\y=\X\beta+\varepsilon$ where $\X$ in an $n\times p$ matrix and $\y\in\R^n$. In this study we used $p=10$ covariates with $n=2000$ observations. The covariates were sampled independently of each other, from a normal distribution with mean zero and variance one. For each of nine covariates, the observations were sampled independently from one another. The remaining covariate, without loss of generality the first covariate, follows an autoregressive scheme with mean lifetime 10, i.e.\ $\Cov(\X_{i1},\X_{j1})=\pi^{|i-j|}$ for $\pi=\exp(-1/10)$. The $p$ individual effects $\beta=(\beta_1,..,\beta_p)$ were sampled independently from each other and independently from the covariates, from a normal distribution with mean zero and variance one. The residual noise $\varepsilon$ was sampled independently from the covariates and effects, from a normal distribution with mean zero and variance $\sigma^2p$, such that $1/\sigma^2$ can be said to be the \textit{signal-to-noise ratio}. The residual noise follows an autoregressive scheme with mean lifetime 10, i.e. $\Cov(\varepsilon_i,\varepsilon_j)=\sigma^2p\rho^{|i-j|}$ for $\rho=\exp(-1/10)$. We then computed the standard ridge estimator $\hat\beta=(\X^T\X+\lambda\I)^{-1}\X^T\y$, the 2REG ridge estimator $\hat\beta=(\X^T\X+\lambda\X^T\X\hat\C_{\mu,\kappa})^{-1}\X^T\y$ and the correctly specified 2REG ridge estimator $\hat\beta=(\X^T\X+\lambda\X^T\X\C)^{-1}\X^T\y$ for various values of $\lambda\geq 0$ and $0\leq \mu,\kappa\leq 1$. The 2REG ridge estimator uses a covariance estimator $\hat\C_{\mu,\kappa}$ which was computed using the block-bootstrap approach outlined in section \ref{section_cov} with $B=2000$ iterations of resampling $\Omega=20$ continuous non-overlapping equal-sized blocks. We used regularization with $\mu,\kappa$ as in equation (\ref{section_cov_mu_kappa}) using prior covariance proportional to $(\X^T\X)^{-1}$. For each resulting estimator $\hat\beta$ we then computed its squared estimation error $(\hat\beta_1-\beta_1)^2+...+(\hat\beta_p-\beta_p)^2$, e.g.\ for the null estimator $\hat\beta=0$ we expect squared estimation error of $p$. Notice that as the columns of $\X$ and independent and have unit variance, the squared estimation error as above is equal to the expected out of sample prediction error. We repeated this process $N=50000$ independent times, and we report the average and standard error of squared estimation errors over the $N$ tries.
\nnline

We first show the results for $\sigma^2=10$ (table 1), a low signal-to-noise environment, where regularization is imperative. We report results for OLS, and for standard and 2REG ridge with optimal $\lambda\geq 0$, selected as the value of $\lambda$ which gave the lowest squared estimation error averaged over the $N$ tries.

\begin{center}
\begin{threeparttable}
\caption{Estimation error with autocorrelation, low signal-to-noise ($\sigma^2=10$)}
    \begin{tabular}{c||c|c|c|c|c}
        \Xhline{2\arrayrulewidth}
         & $\mu=0$ & $\mu=0.2$ & $\mu=0.4$ & $\mu=0.6$  & $\mu=0.8$ \\
        \textbf{method} & $\kappa=0$ & $\kappa=0.2$ & $\kappa=0.4$ & $\kappa=0.6$  & $\kappa=0.8$ \\ \Xhline{2\arrayrulewidth}
        OLS & 0.953 & \scriptsize N/R & \scriptsize N/R & \scriptsize N/R & \scriptsize N/R \\
        $\lambda=0$ &  (0.0033) &  & & & \\ \hline
        standard ridge & 0.869 & \scriptsize N/R & \scriptsize N/R & \scriptsize N/R & \scriptsize N/R \\
        optimal $\lambda$ & (0.0028) & & & & \\ \hline
        2REG ridge & 0.784 & 0.782 & 0.795 & 0.819 & 0.846 \\
        optimal $\lambda$ & (0.0024) & (0.0024) & (0.0024) & (0.0026) &  (0.0027)\\ \hline
        correctly specified & 0.760 & \scriptsize N/R & \scriptsize N/R & \scriptsize N/R & \scriptsize N/R \\
        2REG ridge & (0.0023) &  &  &  &  \\
        optimal $\lambda$ &  &  &  &  &  \\
        \Xhline{2\arrayrulewidth}
    \end{tabular}
    \begin{tablenotes}
      \footnotesize
      \item The average (standard error) of squared estimation error over $N=50000$ tries is reported. N/R - not relevant.
    \end{tablenotes}
\end{threeparttable}
\end{center}

Here, using 2REG ridge over standard ridge doubles the contribution gained by using standard ridge over OLS. This simple setting allows for covariance shrinkage even as light as $\mu=\kappa=0$ with near-optimal result. Recall that 2REG ridge with $\kappa=1$ is identical to standard ridge. In the same experiment, we report the average $\hat\beta_1^2$ over all repetitions, and show that it is more harshly diminished in 2REG ridge than other elements of $\hat\beta$, as expected (see table 2).
\begin{center}
\begin{threeparttable}
\caption{Magnitude of $\hat\beta_1$ with autocorrelation, $\sigma^2=10$}
    \begin{tabular}{c||c|c|c|c|c}
        \Xhline{2\arrayrulewidth}
         & $\mu=0$ & $\mu=0.2$ & $\mu=0.4$ & $\mu=0.6$  & $\mu=0.8$ \\
        \textbf{method} & $\kappa=0$ & $\kappa=0.2$ & $\kappa=0.4$ & $\kappa=0.6$  & $\kappa=0.8$ \\ \Xhline{2\arrayrulewidth}
        OLS & 1.503 & \scriptsize N/R & \scriptsize N/R & \scriptsize N/R & \scriptsize N/R \\
        $\lambda=0$ &  (10.948) & & & & \\ \hline
        standard ridge & 1.238 & \scriptsize N/R & \scriptsize N/R & \scriptsize N/R & \scriptsize N/R \\
        optimal $\lambda$ & (9.037) & & & & \\ \hline
        2REG ridge & 0.694 & 0.808 & 0.931 & 1.053 & 1.160 \\
        optimal $\lambda$ & (9.223) & (9.089) & (9.026) & (9.010) &  (9.020)\\ \hline
        correctly specified & 0.643 & \scriptsize N/R & \scriptsize N/R & \scriptsize N/R & \scriptsize N/R \\
        2REG ridge & (9.163) &  &  &  &  \\
        optimal $\lambda$ & &  &  &  &  \\
        \Xhline{2\arrayrulewidth}
    \end{tabular}
    \begin{tablenotes}
      \footnotesize
      \item The average of $\hat\beta_1^2$ (in parentheses $\hat\beta_1^2+...+\hat\beta_{p}^2$, $p=10$) over $N=50000$ tries is reported. N/R - not relevant.
    \end{tablenotes}
\end{threeparttable}
\end{center}

In a relatively higher signal-to-noise environment, $\sigma^2=2$ (table 3), the contribution of 2REG ridge is smaller in absolute value, but still retains the relative benefit over using standard ridge.
\begin{center}
 \begin{threeparttable}
\caption{Estimation error with autocorrelation, high signal-to-noise ($\sigma^2=2$)}
       \begin{tabular}{c||c|c|c|c|c}
        \Xhline{2\arrayrulewidth}
         & $\mu=0$ & $\mu=0.2$ & $\mu=0.4$ & $\mu=0.6$  & $\mu=0.8$ \\
        \textbf{method} & $\kappa=0$ & $\kappa=0.2$ & $\kappa=0.4$ & $\kappa=0.6$  & $\kappa=0.8$ \\ \Xhline{2\arrayrulewidth}
        OLS & 0.1907 & \scriptsize N/R & \scriptsize N/R & \scriptsize N/R & \scriptsize N/R \\
        $\lambda=0$ &  (0.0007) & & & & \\ \hline
        standard ridge & 0.1871 & \scriptsize N/R & \scriptsize N/R & \scriptsize N/R & \scriptsize N/R \\
        optimal $\lambda$ & (0.0006) & & & & \\ \hline
        2REG ridge & 0.1821 & 0.1818 & 0.1828 & 0.1844 & 0.1860 \\
        optimal $\lambda$ & (0.0006) & (0.0006) & (0.0006) & (0.0006) &  (0.0006)\\ \hline
        correctly specified & 0.1805 & \scriptsize N/R & \scriptsize N/R & \scriptsize N/R & \scriptsize N/R \\
        2REG ridge & (0.0006) &  &  &  &  \\
        optimal $\lambda$ &  &  &  &  &  \\
        \Xhline{2\arrayrulewidth}
    \end{tabular}
       \begin{tablenotes}
      \footnotesize
      \item The average (standard error) of squared estimation error over $N=50000$ tries is reported. N/R - not relevant.
    \end{tablenotes}
\end{threeparttable}
\end{center}
\nnline

In the second study, one of the covariates is associated with a random effect. The setting in this study was identical to that of the first study, except in the following respects. 
There is no autocorrelation in the covariates and response, i.e.\ $\pi=\rho=0$. Instead, additional noise is added with magnitude proportional to the first covariate, recall proposition \ref{section_time_series_prop_random1} and the discussion preceding it. That is, we have $\y=\X\beta+\eta$ where $\eta_i=\varepsilon_i+\X_{i1}b_i$ for $i=1,..,n$, and it remains to describe $b_1,..,b_n$. Indeed, $b_1,..,b_n$ were sampled randomly from a normal distribution with mean zero following an autoregressive scheme such that $\Cov(b_i,b_j)=\Var(b_1)\tau^{|i-j|}$, compare with proposition \ref{section_time_series_prop_random1}. We used $\tau=exp(-1/100)$ to model a relatively-slowly moving effect, alongside $\Var(b_1)=5$ and $\sigma^2=0.5$, recall $\Var(\varepsilon_1)=\sigma^2p=5$, in order for the noise attributed to the random effect to be a discernible part of the total noise $\eta$, recall proposition \ref{section_time_series_prop_random1}. The results are reported in table 4.

\begin{center}
\begin{threeparttable}
\caption{Estimation error with random effect}
    \begin{tabular}{c||c|c|c|c|c}
        \Xhline{2\arrayrulewidth}
         & $\mu=0$ & $\mu=0.2$ & $\mu=0.4$ & $\mu=0.6$  & $\mu=0.8$ \\
        \textbf{method} & $\kappa=0$ & $\kappa=0.2$ & $\kappa=0.4$ & $\kappa=0.6$  & $\kappa=0.8$ \\ \Xhline{2\arrayrulewidth}
        OLS & 0.506 & \scriptsize N/R & \scriptsize N/R & \scriptsize N/R & \scriptsize N/R \\
        $\lambda=0$ &  (0.0030) & & & & \\ \hline
        standard ridge & 0.483 & \scriptsize N/R & \scriptsize N/R & \scriptsize N/R & \scriptsize N/R \\
        optimal $\lambda$ & (0.0027) & & & & \\ \hline
        2REG ridge & 0.367 & 0.375 & 0.396 & 0.430 & 0.464 \\
        optimal $\lambda$ & (0.0020) & (0.0021) & (0.0021) & (0.0023) &  (0.0025)\\ \hline
        correctly specified & 0.362 & \scriptsize N/R & \scriptsize N/R & \scriptsize N/R & \scriptsize N/R \\
        2REG ridge & (0.0020) &  &  &  &  \\
        optimal $\lambda$ &  &  &  &  &  \\
        \Xhline{2\arrayrulewidth}
    \end{tabular}
    \begin{tablenotes}
      \footnotesize
      \item The average (standard error) of squared estimation error over $N=50000$ tries is reported. N/R - not relevant.
    \end{tablenotes}
\end{threeparttable}
\end{center}

\section{Real Data Study}\label{section_real_data}

In this section we describe the results of a real data study examining the performance of the 2REG ridge estimator on a real-world data set of NYSE stock prices.
\nnline

The data was based on NYSE closing prices of five leading technology stocks; MSFT, AAPL, FB, GOOGL and AMZN between February 8, 2013 and February 7, 2018, obtained from the publicly available Kaggle data set 'S\&P 500 stock prices' by Cam Nugent \cite{kaggle}. The data was partitioned into a training sample, consisting of 982 observations dating up to December 30, 2016, and a test sample consisting of 277 observations dating from January 3, 2017 onward. We fit five linear models with regularization, predicting future prices of each of the five stocks, using past prices of all of the five stocks combined as covariates.
\nnline

The data sets were constructed from the closing prices data, for train and test separately, as follows. We fit five distinct models, each predicting the price return of a different stock ten trading days into the future. Formally, the response for each stock was defined by $\y_i = \log(\p_{i+10}) - \log(\p_i)$ where $\p_i$ is the closing price of the stock at the $i$'th row. All five data sets shared a common set of ten covariates. For each of the five stocks, we included two covariates in $\X$; short (one trading day) and long (five trading day) past price return of that stock. Formally, the short price return covariate was defined by $\x_i^{\text{short}}=\log(\p_i)-\log(\p_{i-1})$ and the long by $\x_i^{\text{long}}=\log(\p_i)-\log(\p_{i-5})$. In these data sets, both $\y$ and the long covariates in $\X$ are autocorrelated due to congruence between proximate price returns. We cannot assume strict exogeneity here, if only for the self price return covariates predicting their future selves. However, contemporaneous exogeneity is usually assumed to apply in this situation \cite{greene}. The fixed or random nature of the effects are unknown. Whenever $\X$ or $\y$ had missing values due to boundary issues, the entire row was dropped, in total the five head rows and ten tail rows were dropped from each data set.
\nnline

We fit a standard ridge estimator $\hat\beta=(\X^T\X+\lambda\I)^{-1}\X^T\y$ and a 2REG ridge estimator $\hat\beta=(\X^T\X+\lambda\X^T\X\hat\C_{\mu,\kappa})^{-1}\X^T\y$ for various values of $\lambda\geq 0$. The 2REG ridge estimator was computed using the block bootstrap approach outlined in section \ref{section_cov}, with $B=2000$ iterations of resampling $\Omega=10$ continuous non-overlapping blocks of roughly equal size. We applied covariance regularization with $\mu,\kappa$ as in equation (\ref{section_cov_mu_kappa}) using prior covariance proportional to $(\X^T\X)^{-1}$. The values of $\mu,\kappa$ were selected from a grid so as to best predict the covariance matrix out-of-sample, in Frobenius norm, recall section \ref{section_cov}. For this, we used cross validation over the same folds previously specified and the same amount of bootstrap iterations were used to compute the covariance matrices. The out-of-sample covariances used for cross validation were not regularized.
\nnline

For each of the estimators we computed its out-of-sample predictions on the test data, and we report $r^2$ values over all five data sets predicting different stocks combined. The $r^2$ values were calculated against the null model $\hat\beta=0$. In figure 5, we plot out-of-sample $r^2$ values along different ridge $\lambda$ values, for standard and 2REG ridge.

\begin{figure}[H]
    \centering
    Figure 5: Out-of-sample $r^2$ value on NYSE stock price data
    \includegraphics[scale=0.5]{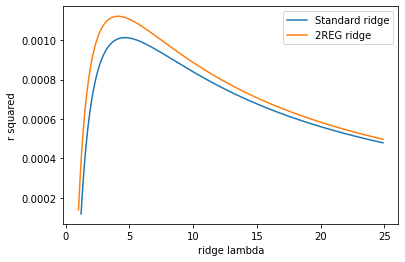}    
\end{figure}

See that standard and 2REG ridge are similarly scaled due to the covariance estimator normalization, recall section \ref{section_cov}. Note that $r^2$ values are very small, as is typical for this kind of data. The standard ridge estimator attains at most $r^2=0.00101$, while 2REG ridge has $r^2=0.00112$ at its peak. Without covariance regularization, 2REG ridge yields negative $r^2$ out of sample.

\begin{appendix}

\section{Proofs for Time Series Examples}\label{appendix_time_series}

In this appendix we provide proofs for the statements made in section \ref{section_time_series}.

\begin{proof}[Proof of lemma \ref{section_time_series_lemma}]
    Denote $\hat\beta=\hat\beta^n$, $\X=(\X_1,..,\X_n)$ and $\varepsilon=(\varepsilon_1,..,\varepsilon_n)$, permitting the abuse in disregarding the dimension $n$. As $\varepsilon$ has mean zero conditioned upon $\X$ we have
    \begin{align}\label{appendix_time_series_lemma_var}
        n\Var(\hat\beta) \ = \ n\cdot \E\bigg[\frac{\X^T\varepsilon\varepsilon^T\X}{(\X^T\X)^2}\bigg] \ = \ \E\bigg[\frac{\X^T\SSigma_n\X/n}{(\X^T\X/n)^2}\bigg].
    \end{align}
  
    By the assumptions of the lemma $(\X^T\X/n)^2\to1$ in probability, whereas the numerator $\{\X^T\SSigma_n\X/n\}_{n=1}^\infty$ in (\ref{appendix_time_series_lemma_var}) converges in probability as well and we shall compute the limit. See that
    \begin{align}\label{appendix_time_series_lemma_numerator}
    \E[\X^T\SSigma_n\X] \ = \ \sum_{i,j=1}^n \E[(\SSigma_n)_{ij}\X_i\X_j] \ = \ \sum_{i,j=1}^n\Cov(\X_i\varepsilon_i,\ \X_j\varepsilon_j)
    \end{align}
    and we denote this quantity by $V_n=\E[\X^T\SSigma_n\X]$.
    As $\{\X_n\}_{n=1}^\infty$ have uniformly bounded fourth moments, the sequence $\{\X^T\X/n\}_{n=1}^\infty$ is uniformly integrable \cite{gut}. Since $\{\SSigma_n\}_{n=1}^\infty$ is uniformly bounded in operator norm there exists $M>0$ such that $\X^T\SSigma_n\X\leq M\X^T\X$, hence the sequence $\{\X^T\SSigma_n\X/n\}_{n=1}^\infty$ is uniformly integrable as well. 
    Therefore, by convergence by uniform integrability \cite{gut} and by (\ref{appendix_time_series_lemma_numerator}), the limit is
    $$
    \frac{\X^T\SSigma_n\X}{n} \ \overset{p}{\longrightarrow} \ \lim_{n\to\infty} \E\bigg[\frac{\X^T\SSigma_n\X}{n}\bigg] \ = \ \lim_{n\to\infty} \frac{1}{n} V_n.
    $$
    Consequently, the integrand in (\ref{appendix_time_series_lemma_var}) converges in probability to
    \begin{align}\label{appendix_time_series_lemma_plim}
    \frac{\X^T\SSigma_n\X/n}{(\X^T\X/n)^2} \ \overset{p}{\longrightarrow} \ \lim_{n\to\infty} \frac{1}{n}V_n.
    \end{align}
    It remains to apply convergence by uniform integrability again, this time on the variance formula (\ref{appendix_time_series_lemma_var}) to obtain that the limit variance is the probability limit of the integrand. Indeed, using $M$ as before we have
    \begin{align*}
        \frac{\X^T\SSigma_n\X/n}{(\X^T\X/n)^2} \ \leq \ \frac{M}{\X^T\X/n}.
    \end{align*}
    By assumption, $\{ n(\X^T\X)^{-1}\}_{n=1}^\infty$ are uniformly integrable hence the integrand in (\ref{appendix_time_series_lemma_var}) constitutes a uniformly integrable sequence as well, and thus by convergence by uniform integrability and by (\ref{appendix_time_series_lemma_var}) and (\ref{appendix_time_series_lemma_plim}) we obtain
    $$
    n\Var(\hat\beta) \ = \ \E\bigg[\frac{\X^T\SSigma_n\X/n}{(\X^T\X/n)^2}\bigg] \ \longrightarrow \ \lim_{n\to\infty} \frac{1}{n}V_n
    $$
    as required.
\end{proof}

\begin{proof}[Proof of proposition \ref{section_time_series_prop_ar}]
    Following lemma \ref{section_time_series_lemma}, see that
    \begin{align*}
        \Cov(\X_i\varepsilon_i,\ \X_j\varepsilon_j) \ = \ \E[\X_i\X_j\varepsilon_i\varepsilon_j] \ = \ \E[\X_i\X_j\E[\varepsilon_i\varepsilon_j|\X_i\X_j]] \ = \ \sigma^2(\pi\rho)^{|i-j|}.
    \end{align*}
    Hence,
    $$
    \frac{1}{n}\sum_{i,j=1}^n\Cov(\X_i\varepsilon_i,\ \X_j\varepsilon_j) \ = \ \frac{\sigma^2}{n}\sum_{i,j=1}^n(\pi\rho)^{|i-j|} \ = \ \frac{\sigma^2}{n}\sum_{\omega=-n}^n (n-|\omega|)(\pi\rho)^{|\omega|}.
    $$
    We now show that for any $0\leq a<1$ we have
    \begin{align}\label{appendix_time_series_prop_ar_full}
    \sum_{\omega=-n}^n\bigg(1-\frac{|\omega|}{n}\bigg) a^{|\omega|} \ \overset{n\to\infty}{\longrightarrow} \ \frac{1+a}{1-a}
    \end{align}
    and by lemma \ref{section_time_series_lemma} this will suffice to prove the proposition. Indeed,
    \begin{align}\label{appendix_time_series_prop_ar_part1}
        \sum_{\omega=-n}^n a^{|\omega|} \ = \ -1 + 2\sum_{\omega=0}^n a^\omega \ = \ -1 + \frac{2(1-a^{n+1})}{1-a} \ \overset{n\to\infty}{\longrightarrow} \ \frac{1+a}{1-a}.
    \end{align}
    Likewise,
    \begin{align}\label{appendix_time_series_prop_ar_part2}
        \frac{1}{n}\sum_{\omega=-n}^n |\omega|a^{|\omega|} \ \leq \ \frac{2}{n}\sum_{\omega=0}^\infty\omega a^\omega \ = \ \frac{2a}{n(1-a)^2} \ \overset{n\to\infty}{\longrightarrow} \ 0.
    \end{align}
    Combining (\ref{appendix_time_series_prop_ar_part1}) with (\ref{appendix_time_series_prop_ar_part2}) we get (\ref{appendix_time_series_prop_ar_full}).
\end{proof}

\begin{proof}[Proof of proposition \ref{section_time_series_prop_random1}]
    By lemma \ref{section_time_series_lemma},
    \begin{align}\label{section_time_series_prop_random1_goal}
        \lim_{n\to\infty} n\Var(\hat\beta_n) \ = \ \lim_{n\to\infty} \frac{1}{n}\sum_{i=1}^\infty \Cov\big(\X_i(\varepsilon_i+\X_ib_i), \ \X_j(\varepsilon_j +\X_jb_j)\big).
    \end{align}
    Denote $V_{ij}=\Cov(\X_i(\varepsilon_i+\X_ib_i), \ \X_j(\varepsilon_j +\X_jb_j))$.
    As all variables have mean zero and as $\{b_i\}_{i=1}^\infty$, $\{\X_i\}_{i=1}^\infty$, $\{\varepsilon_i\}_{i,j=1}^\infty$ are independent of each other, we have
    \begin{align*}
        V_{ij} \ &=  \ \E[\X_i\X_j(\varepsilon_i+\X_ib_i)(\varepsilon_j+\X_jb_j)]
        \\ 
        &= \ \E[\X_i\X_j]\cdot\E[\varepsilon_i\varepsilon_j] \ + \ \E[\X_i^2\X_j^2]\cdot\E[b_ib_j]
        \\
        &= \ \delta_{ij}\cdot(\sigma^2 + \E[\X_i^4]\cdot\Var(b_i)) \ + \ (1-\delta_{ij})\cdot\Cov(b_i,b_j)
    \end{align*}
    where $\delta_{ij}$ equals one when $i=j$ and zero otherwise. As $\{\X_i\}_{i=1}^\infty$ are identically distributed we have $\E[\X_i^4]=\E[\X_1^4]$ for all $i=1,..,n$. Furthermore, $\Cov(b_i,b_j)=\Var(b_1)\tau^{|i-j|}$ and in particular $\Var(b_i)=\Var(b_1)$. Then, by (\ref{section_time_series_prop_random1_goal}) we get
    $$
    \lim_{n\to\infty} n\Var(\hat\beta_n) \ = \ \sigma^2 + \Var(b_1)\cdot\bigg(\E[\X_1^4]-1 +\lim_{n\to\infty} \frac{1}{n} \sum_{i,j=1}^n \tau^{|i-j|}\bigg).
    $$
    In the proof of proposition \ref{section_time_series_prop_ar} we computed
    $$
    \lim_{n\to\infty} \frac{1}{n} \sum_{i,j=1}^n \tau^{|i-j|} \ = \ \frac{1+\tau}{1-\tau}
    $$
    and this concludes the proof of the proposition.
\end{proof}

\begin{proof}[Proof of proposition \ref{section_time_series_prop_random2}]
     By lemma \ref{section_time_series_lemma},
    \begin{align}\label{section_time_series_prop_random2_goal}
        \lim_{n\to\infty} n\Var(\hat\beta_n) \ = \ \lim_{n\to\infty} \frac{1}{n}\sum_{i,j=1}^\infty \Cov\big(\X_i(\varepsilon_i+\Z_ib_i), \ \X_j(\varepsilon_j +\Z_jb_j)\big).
    \end{align}
    Denote $V_{ij}=\Cov(\X_i(\varepsilon_i+\X_ib_i), \ \X_j(\varepsilon_j +\X_jb_j))$.
    As all variables have mean zero and as $\{b_i\}_{i=1}^\infty$, $\{\X_i\}_{i=1}^\infty$, $\{\Z_i\}_{i=1}^\infty$, $\{\varepsilon_i\}_{i=1}^\infty$ are independent of each other, we have
    \begin{align*}
        V_{ij} \ &=  \ \E[\X_i\X_j(\varepsilon_i+\Z_ib_i)(\varepsilon_j+\Z_jb_j)]
        \\ 
        &= \ \E[\X_i\X_j]\cdot\E[\varepsilon_i\varepsilon_j] \ + \ \E[\X_i\X_j]\cdot\E[\Z_i\Z_j]\cdot\E[b_ib_j]
        \\
        &= \ \delta_{ij}\cdot(\sigma^2\ +\ \Var(b_1)).
    \end{align*}
    where $\delta_{ij}$ equals one when $i=j$ and zero otherwise. By (\ref{section_time_series_prop_random2_goal}), the proposition is proven.
\end{proof}

\section{Monotonicity of Ridge}\label{appendix_monotone}

In this appendix we prove proposition \ref{section_2reg_prop_monotone}.

\begin{proof}[Proof of proposition \ref{section_2reg_prop_monotone}]
    See that $\hat\beta_{\lambda}  = (\I+\lambda\LLambda)^{-1}\ols$ and consequently
    \begin{align}\label{appendix_monotone_cov}
        \Cov(\hat\beta_{\lambda}) \ = \ (\I+\lambda\LLambda)^{-1}\C(\I+\lambda\LLambda^T)^{-1}
    \end{align}
    where $\C=\Cov(\ols)$. The covariance $\Cov(\hat\beta_{\lambda})$ is monotone decreasing in $\lambda\geq 0$ in the sense of positive definite matrices, if and only if its inverse $\Cov(\hat\beta_{\lambda})^{-1}$ is increasing \cite{horn_johnson}. Inverting, we get
\begin{align}\label{appendix_monotone_cov_inv}
    \Cov(\hat\beta_{\lambda})^{-1} \ = \ \C^{-1} \ + \ \lambda\big(\LLambda^T\C^{-1} + \C^{-1}\LLambda\big) \ + \ \lambda^2\LLambda^T\C^{-1}\LLambda.
\end{align}

The matrix $\LLambda^T\C^{-1}\LLambda$ in (\ref{appendix_monotone_cov_inv}) is symmetric and positive semidefinite, as $\C$ is. If the symmetric matrix $\LLambda^T\C^{-1}+\C^{-1}\LLambda$ is positive semidefinite as well, then $\Cov(\hat\beta_{\lambda})^{-1}$ is monote increasing for as $\lambda$ increases, a positive semidefinite matrix is added. Conversely,
$$
\frac{d}{d\lambda}\bigg|_{\lambda=0}\Cov(\hat\beta_{\lambda})^{-1} \ = \ \LLambda^T\C^{-1}+\C^{-1}\LLambda
$$
therefore whenever $\LLambda^T\C^{-1}+\C^{-1}\LLambda$ is not positive semidefinite, $\Cov(\hat\beta_{\lambda})^{-1}$ cannot be monotone increasing.
\end{proof}

\end{appendix}

\end{document}